\documentclass[12pt, a4paper, reqno]{amsart}
\usepackage{mathrsfs}
\usepackage{bbm}
\usepackage{amsmath,amscd,amssymb,latexsym}
\usepackage[all]{xy}

\input xypic

\evensidemargin=\oddsidemargin
\DeclareMathVersion{can}
\DeclareMathAlphabet{\can}{OT1}{cmss}{m}{n}
\newtheorem{thm}{Theorem}[section]

\newtheorem{lem}[thm]{Lemma}

\newtheorem{exa}[thm]{Example}
\theoremstyle{definition}

\theoremstyle{fact}

\theoremstyle{conjecture}

\numberwithin{equation}{section}

\newcommand{\ord}{\operatorname{ord}}

\begin{document}
\title[Further factorization of $x^{n}-1$ over a finite field] {Further factorization of $x^n-1$ over a finite field}
\author[Y. Wu]{Yansheng Wu}
\address{\rm Department of Mathematics, Nanjing University of Aeronautics and Astronautics,
Nanjing, 210016, P. R. China; State Key Laboratory of Cryptology, P. O. Box 5159, Beijing, 100878, PR China}
 \email{wysasd@163.com}

\author[Q. Yue]{Qin Yue}
\address{\rm Department of Mathematics, Nanjing University of Aeronautics and Astronautics,
Nanjing, 210016, P. R. China; State Key Laboratory of Cryptology, P. O. Box 5159, Beijing, 100878, PR China}
\email{yueqin@nuaa.edu.cn}
\author[S. Fan]{Shuqin Fan}
 \address{\rm State Key Laboratory of Cryptology, P. O. Box 5159, Beijing, 100878, PR China}
\email{fansq@sklc.org}

\thanks{The paper is supported by National Natural Science Foundation of China (No. 61772015, 11601475, 11661014), the Guangxi  Science Research and Technology Development Project
(1599005-2-13), and
Foundation of Science and Technology on Information Assurance Laboratory (No. KJ-15-009).}

\subjclass[2010]{  11T06, 12E05}

\keywords{ Irreducible factor, cyclotomic polynomials}

\begin{abstract} Let $\Bbb F_q$ be a finite field with $q$ elements and $n$ a positive integer.  Mart\'{\i}nez,  Vergara and Oliveira \cite{MVO} explicitly factorized  $x^{n} - 1$ over $\Bbb F_q$ under the condition of $rad(n)|(q-1)$. In this paper, suppose that $rad(n)\nmid (q-1)$ and $rad(n)|(q^w-1)$,
 where $w$ is a prime,  we explicitly factorize  $x^{n}-1$ into irreducible factors in $\Bbb F_q[x]$ and count the number of its irreducible factors.
\end{abstract}
\maketitle

\section{Introduction}

Let $\Bbb F_q$ be a finite field of order $q$, where $q$ is a positive power of a prime $p$.
Factorization of  polynomials over finite fields is a classical topic of mathematics. There are various computational problems depending in one way or another on the factorization of polynomials over finite fields. The factorization of $x^{n} - 1$ has a very close relation to the factorization of the $n$-th cyclotomic polynomial $\Phi_n(x)$ (see \cite{LN}), which has received a good deal of attention.

In 2007, the explicit factorization of  $\Phi_{2^{n}r}(x )$  over a finite field $\Bbb F_{q}$ was studied by  Fitzgerald and  Yucas \cite{F},  where $r$ is an odd prime with $q \equiv \pm1\pmod r$.
This gave the explicit irreducible factors of $\Phi_{2^{n}3}(x )$ and the Dickson polynomial $D_{2^{n}3}(x )$ completely.

In 2012,  Wang and Wang \cite{WW} showed that  all irreducible factors of $2^{n}r$-th cyclotomic polynomials can be obtained easily from irreducible factors of cyclotomic polynomials of small orders. Hence, the explicit factorization of $2^{n}5$-th cyclotomic polynomials over finite fields was obtained by this way.
Assuming that the explicit factors of  $\Phi_{r}(x )$ are known,  Tuxanidy and  Wang obtained the irreducible factors of  $\Phi_{2^{n}r}(x )$ over $\Bbb F_{q}$ concretely in \cite{TW}, where $r \geqslant3$ is an arbitrary odd integer.

In 2013,  Chen,  Li and  Tuerhong \cite{CLT} obtained the irreducible factorization of $x^{2^mp^n}-1$ over $\Bbb F_q$  in a very explicit form, where $p$ is an odd prime divisor of $q-1$. It is also shown that all the irreducible factors of $x^{2^mp^n}-1$ over $\Bbb F_q$ are either binomials or trinomials.

Later in 2015,  Mart\'{\i}nez,  Vergara and Oliveira \cite{MVO} investigated when the polynomial $x^{n} -1$ can be split into irreducible binomials $x^{t} -a$ or trinomials $x^{2t} - ax^{t} + b$ over $\Bbb F_q$. Most recently,   Wu, Zhu,  Feng and Yang \cite{WZF}  showed that all irreducible factors of $p^{n}$-th and $p^{n}r$-th cyclotomic polynomials can be obtained from irreducible factors of cyclotomic polynomials of small orders. Applying this result, they presented a general idea to factorize cyclotomic polynomials over finite fields.
Some explicit factorizations of certain cyclotomic polynomials or Dickson polynomials can be found in \cite{BGM,CL,CLT,F,LY1,LY2,LN,W,WY},  etc.

This paper is organized as follows. In Section 2, we present some basic results. In Section 3, suppose that $rad(n)\nmid (q-1)$ and $rad(n)|(q^w-1)$, where $w$ is prime,  we shall explicitly give the  factorization of the polynomial $x^{n}-1$ into irreducible factors over $\Bbb F_{q}$. Moreover, the total number of irreducible factors of $x^{n}-1$ is given. In Section 4, we
 conclude this paper.



For convenience, we introduce the following notations in this paper:

\begin{tabular}{ll}
$q$ & a power of a prime $p$,   \\
$\Bbb F_{q}$ & finite field $GF(q)$,\\
$Tr_{q/p}$ & trace function from $\Bbb F_{q}$ to $\Bbb F_{p}$,\\
$\pi$ & a generator of $\Bbb F_{q^{2w}}^\ast$, \\
$\delta$ & a generator of $\Bbb F_{q^{w}}^\ast$ satisfying $\delta=\pi^{q^{w}+1}$,\\
$\alpha$ & a generator of $ \Bbb F_{q^{2}}^{\ast}$ satisfying $\alpha= \pi^{\frac{q^{2w}-1}{q^{2}-1}}$,\\
$\theta$ & a generator of $\Bbb F_q^\ast$ satisfying $\theta= \pi^{\frac{q^{2w}-1}{q^{}-1}}$,\\


$m_{s}$ & $m_{s}:=\frac{n}{\gcd(n,q^{s}-1)}$, where $s$ is a positive integer,\\
$l_{s}$ & $l_{s}:=\frac{q^{s}-1}{\gcd(n,q^{s}-1)}$, where $s$ is a positive integer,\\

\end{tabular}

 \section{Preliminaries}


Let  $n = p_{1}^{\alpha_{1}}p_{2}^{\alpha_{2}} \cdots p_{l}^{\alpha_{l}}$ be  the prime factorization, where $p_1,\ldots, p_l$ are distinct primes and positive integers $\alpha_i\ge 1$ for $1\le i\le l$. We denote $rad(n) = p_{1} p_{2} \cdots p_{l}$ and  $v_{p_i}(n)=\alpha_i$, $1\le i\le l$.
 For any two integers $m$ and $n$ with $\gcd(m,n)=1$, $\ord_{n} (m)$  denote  the minimum positive integer $k$ such that $m^k \equiv 1 \pmod n$.

There is a classical remarkable criterion on irreducible binomials over $\Bbb F_q$, which was given by Serret in 1866.
\begin{lem} \cite[Theorem 3.75]{LN}
Assume that $t\geq 2$ is a positive integer and $\eta \in \Bbb F_q^\ast$. Then the binomial
 $x^t-\eta$ is irreducible over $\Bbb F_q$ if and only if both the following conditions are satisfied:

  (i) $ rad(t)$ divides $o(\eta)$, where $o(\eta)$ denotes the order of $\eta$ in  $\Bbb F_{q}^{\ast}$,

  (ii) $\gcd(t, \frac{q-1}{o(\eta)})=1$,

(iii) if $4\mid t$, then $4 | (q-1)$.
\end{lem}

In fact, if $t=4s$, $q\equiv 3\pmod 4$ and $a$ is not a square in $\Bbb F_q$, then $x^t-a=(x^{2s}-bx^s+c)(x^{2s}+bx^s+c)$ is reducible over $\Bbb F_q$, where $b, c\in \Bbb F_q$, $c^2=-a$, $b^2=2c$.

We need the following results concerning the explicit factorization of the polynomial $x^{n}-1$ over $ \Bbb F_{q}$.


\begin{lem}\cite[Corollary 1]{MVO} Let $ \Bbb F_q$ be a finite field and $n$ a positive integer such that both $rad(n)|(q-1)$ and
either
 $q \not\equiv 3 \pmod {4}$ or $8\nmid n$. Set $m_1=\frac {n}{\gcd(n, q-1)}$, $l_1=\frac{q-1}{\gcd(n, q-1)}$, and  $\theta$  a  generator of $\Bbb F_q^*$. Then
the factorization of $x^{n} -1$ into irreducible factors in $\Bbb F_q [x ]$ is

$$\prod\limits_{t\mid m_{1}} \prod_{\mbox{\tiny$
\begin{array}{c}
1\leqslant u\leqslant \gcd(n,q-1)\\
\gcd(u,t)=1\\
\end{array}$
}}(x^{t}-\theta^{ul_{1}}).$$

Moreover,
for each $t| m_{1}$, the number of irreducible factors of degree $t$ is $\frac{\varphi(t)}{t}\cdot \gcd(n,q-1)$, where $\varphi$ denotes the Euler Totient function,  and the total number of irreducible factors is
$$\gcd(n,q-1)\cdot \prod_{\mbox{\tiny$
\begin{array}{c}
p\mid m_{1}\\
p~ prime\\
\end{array}$
}}\bigg(1+v_{p}(m_{1})\cdot \frac{p-1}{p}\bigg).$$
\end{lem}

\begin{lem}\cite[Corollary 2]{MVO} Let $ \Bbb F_q$ be a finite field and $n$ a positive integer such that  $rad(n)|(q-1)$,
 $q \equiv3 \pmod 4$, and   $8|n$. Set $m_2=\frac{n}{\gcd(n, q^2-1)}$, $l_2=\frac{q^2-1}{\gcd(n, q^2-1)}$, $l_1=\frac{q-1}{\gcd(n, q-1)}$, $r=\min\{v_2(n/2),v_2(q+1)\}$,  $\alpha$ a  generator of $\Bbb F_{q^2}^*$ satisfying $\theta=\alpha^{q+1}$.
Then the factorization of $x^{n} -1$ into irreducible factors in $\Bbb F_q [x ]$ is

$$\prod_{\mbox{\tiny$
\begin{array}{c}
t| m_{2}\\
t~odd\\
\end{array}$
}} \prod_{\mbox{\tiny$
\begin{array}{c}
1\leqslant v\leqslant \gcd(n,q-1)\\
\gcd(v,t)=1\\
\end{array}$
}}(x^{t}-\theta^{vl_{1}})\cdot\prod\limits_{t\mid m_{2}} \prod\limits_{u\in \mathcal{R}_{t}} (x^{2t}-(\alpha^{ul_{2}}+\alpha^{qul_{2}})x^{t}+\theta^{ul_{2}}), $$ where
$$\mathcal{R}_t=\bigg\{u\in \Bbb N: \begin{array}{l}
1\le u\le \gcd(n,q^{2}-1),\gcd(u,t)=1,
 \\ 2^{r}\nmid u, ~and~ u<\{qu\}_{\gcd(n,q^{2}-1)}\end{array}\bigg\}.$$
Note that  $\{a\}_{b}$ denotes the remainder of the division of $a$ by $b$.

Moreover,  for each $t$ odd with $t| m_{2}$, the numbers of irreducible binomials of degree $t$  is $\frac{\varphi(t)}{t}\cdot \gcd(n,q-1)$; the number of irreducible trinomials of degree $2t$ is

\begin{equation*}\left\{
   \begin{array}{ll}
    \frac{\varphi(t)}{2t}\cdot 2^{r_{}}\gcd(n,q-1), &\mbox{ if $ t$ is even},\\
       \frac{\varphi(t)}{2t}\cdot (2^{r_{}}-1)\gcd(n,q-1), &\mbox{ if $ t$ is odd}.\\
    \end{array}\right.
  \end{equation*}
The total number of irreducible factors of $x^{n}-1$ over $\Bbb F_q$ is

$$\gcd(n,q-1)\cdot \bigg(\frac1 2 +2^{r_{}-2}(2+v_{2}(m_{2}))\bigg)\cdot \prod_{\mbox{\tiny$
\begin{array}{c}
p|m_{2}\\
p~ odd ~prime\\
\end{array}$
}}\bigg(1+v_{p}(m_{2})\frac{p-1}{p}\bigg).$$
\end{lem}

\begin{lem}\cite[Lemma 14]{WZF} Let $n$ be a positive integer, $q$  a prime power coprime with $n$, and $w$  the smallest positive integer such that $q^{w} \equiv 1 \pmod {rad(n)}$. If $f (x)$ is an irreducible factor of $\Phi_{n} (x)$ over $\Bbb F_{q^{w}}$, then $f^{\sigma}(x)$ is also an irreducible factor of $\Phi_{n} (x)$ over $\Bbb F_{q^{w}}$, where $\sigma : \Bbb F_{q^{w}} \longrightarrow
\Bbb F_{q^{w}} ; \alpha \longmapsto \alpha ^{q}$, is the Frobenius mapping of $\Bbb F_{q^{w}} $ over $\Bbb F_{q} $. Moreover, $f(x)f^{\sigma}(x)f^{\sigma^{2}}(x)\cdots f^{\sigma^{w-1}}(x)$ is an irreducible factor of  $\Phi_{n} (x)$ over $\Bbb F_{q} $.
\end{lem}

\section{Factorization of $x^{n}-1$ over $\Bbb F_q$}

In this section, we always assume that $\gcd(n,q)=1$ and the smallest positive integer $w$ satisfying $q^{w} \equiv 1 \pmod {rad(n)}$ is a prime.

For each positive  divisor $d$ of $n$,  the order of $q$ modulo $rad(d)$ is 1 or $w$ by $rad(d)| rad(n)$ and $w$ a prime.


Suppose that   $x^{n}-1=\prod\limits_{d|n}\Phi_{d}(x)=\prod\limits_{d|n}f_{d,1}(x)f_{d,2}(x)\cdots f_{d,s_{d}}(x)$ is the irreducible factorization of $x^{n}-1$ over $\Bbb F_{q^w}$. We define the set  $$D=\{d\mid n: 1\le  d\le n, rad(d)|(q-1)\}.$$ Hence those irreducible polynomials can be divided into two parts:
$$X^{(1)}=\{f_{d,j}(x):d\in D, 1\leqslant j\leqslant s_{d}\}, X^{(2)}=\{f_{d,j}(x):d\notin D, 1\leqslant j\leqslant s_{d}\}.$$

 From \cite{WZF}, there is an equivalence relation $\sim$ on $X^{(2)}$ as follows:   $$f_{d,i}(x)\sim f_{d,j}(x)~ {if}~ and ~only~ if~ f_{d,j}(x)= f_{d,i}^{\sigma^{k}}(x) ~for ~some~ 0 \leqslant k\leqslant w - 1. $$
By  Lemma 2.4, we have the irreducible factorization of $x^{n}-1$ over $\Bbb F_{q}$

$$x^{n}-1=\prod\limits_{g(x)\in X^{(1)} }g(x) \prod\limits_{f(x)\in X^{\ast} } f(x)f^{\sigma}(x)f^{\sigma^{2}}(x)\cdots f^{\sigma^{w-1}}(x),$$
 where $X^{\ast}$ is the complete system  of equivalence class representatives of $X^{(2)}$ relative to $\sim$.
\begin{lem} Let $w$ be a prime.

$(1)$ Then we have  $$\gcd(q-1, \frac{q^w-1}{q-1})=\left\{\begin{array}{ll}
1, &\mbox{ if $w\nmid (q-1)$,}\\
w, &\mbox{ if $w|(q-1)$.
}\end{array}\right.$$

$(2)$ If $w$ is  odd, then $v_w(\frac{q^w-1}{q-1})=1$.

$(3)$ If $w\nmid(q-1)$, then $w\nmid\frac{q^w-1}{q-1}$.
\end{lem}
 \begin{proof} (1) Since $\frac{q^{w}-1}{q-1}=1+q+q^{2}+\cdots+q^{w-1}=w+(q-1)+(q^{2}-1)+\cdots+(q^{w-1}-1)$, $\gcd(q-1,\frac{q^{w}-1}{q-1})=\gcd(q-1, w)=1$ or $w$.

 (2) If $w$ is an odd prime and $v_w(q-1)=r\ge 1$, i.e. $q=1+aw^r$, $\gcd (a, w)=1$. Then $q^w-1\equiv aw^{r+1} \pmod {w^{r+2}}$ and $v_w(\frac{q^w-1}{q-1})=1$.

 (3) Suppose that $q-1=aw+b, $ where $a,b$ are integers and $ 0<b<w$.  Then $\frac{q^{w}-1}{q-1}\equiv 1+(b+1)+\cdots+(b+1)^{w-1}\equiv b^{w-1}\pmod w$.
 As $w$ is a prime and $b\neq0$, we have $v_w(\frac{q^w-1}{q-1})=0$.
 \end{proof}

  \subsection{ $w$ is an odd prime.} In the subsection, we always assume  that $\ord_{rad(n)}(q)=w$ is  an odd prime.

First, we consider the case:  either $q\not\equiv3\pmod 4$ or $8\nmid n$.
 Let $m_w=\frac{n}{\gcd(n,q^w-1)}$, $l_w=\frac{q^w-1}{\gcd(n,q^w-1)}$, and  $\delta$ a  generator of $\Bbb F_{q^w}^*$. By Lemma 2.2, there is a irreducible factorization of $x^n-1$ over $\Bbb F_{q^{w}}$:

 \begin{equation}
 x^{n} -1=\prod\limits_{t|m_w} \prod_{\mbox{\tiny$
\begin{array}{c}
1\leqslant u\leqslant \gcd(n,q^{w}-1)\\
\gcd(u,t)=1\\
\end{array}$
}}(x^{t}-\delta^{ul_w}).
\end{equation}

  Next we will explicitly give the irreducible factorization of $x^{n}-1$ over $\Bbb F_{q}$. In fact, we only investigate whether $\delta^{ul_w}\in \Bbb F_q$.

\begin{thm}  Suppose that $\ord_{rad(n)}(q)=w$, where $w$ is an odd prime, and either $q\not\equiv3 \pmod 4$ or $8\nmid n$.   Let $n=w^{v_w(n)}n_1n_2$, $rad(n_1)|(q-1)$, $rad(n_2)|\frac{q^w-1}{q-1}$, $\gcd(w,n_1n_2)=1$.   Set $m_w=\frac{n}{\gcd(n,q^w-1)}$, $l_w=\frac{q^w-1}{\gcd(n,q^w-1)}$, $m_{w,1}=\frac{n_1}{\gcd(n_1, q-1)}$, $l_{1}=\frac{q-1}{\gcd(n, q-1)}$,  $\delta$ a  generator of $\Bbb F^{\ast}_{q^w}$ and  $\theta=\delta^{\frac{q^w-1}{q-1}}$. Then

$(1)$ The irreducible factorization of  $x^{n}-1$  over $\Bbb F_{q}$  is

 \begin{equation}\prod_{\mbox{\tiny$
\begin{array}{c}
 t|m_{w,1}\\
\end{array}$
}}
 \prod_{\mbox{\tiny$
\begin{array}{c}
1\le u'\le \gcd(n, q-1)\\
\gcd(u', t)=1
\end{array}$
}}(x^{t}-\theta^{u'l_{1}})\prod_{\mbox{\tiny$
\begin{array}{c}
t|m_w\\ u\in \mathcal{S}_{t}
\end{array}$
}}\prod\limits_{k=0}^{w-1} (x^{t}-\delta^{ul_{w}})^{\sigma^{k}},\end{equation}
where $$\mathcal{S}_{t}=\bigg\{u\in \Bbb N: \begin{array}{l}
1\le u\le \gcd(n,q^{w}-1),\gcd(u,t)=1,
 \\ \frac{q^{w}-1}{q-1}\nmid ul_{w}, u=\min\{u, qu, \cdots, q^{w-1}u\}_{\gcd(n, q^w-1)}\end{array}\bigg\}$$
and $u=\min\{u, qu, \cdots, q^{w-1}u\}_{\gcd(n, q^w-1)}$ denotes that $u$ is   the minimum remainder modulo $\gcd(n, q^w-1)$ in  the set $\{u, qu,\cdots, q^{w-1}u\}$.

$(2)$ For each $t|m_{w,1}$, the number of irreducible factors of degree $t$ and $wt$ are $\frac {\varphi(t)}t\cdot\gcd (n, q-1)$ and  $\frac{\varphi(t)}{wt}\cdot(\gcd(n, q^w-1)-\gcd(n, q-1))$, respectively. For each $t|m_w$ and $t\nmid m_{w,1}$, the number of irreducible factors of degree $wt$ is $\frac{\varphi(t)}{wt}\cdot\gcd(n, q^w-1)$. The total number of irreducible factors is

\begin{eqnarray*}&&\prod_{\mbox{\tiny$
\begin{array}{c}
p| m_{w}\\
p~prime\\
\end{array}$
}}(1+v_p(m_{w})\frac {p-1}p)\cdot \frac{\gcd(n, q^w-1)}w\\ &+&\prod_{\mbox{\tiny$
\begin{array}{c}
p| m_{w,1}\\
p~prime\\
\end{array}$
}}(1+v_p(m_{w,1})\frac{p-1}p)\cdot \frac{w-1}{w}\cdot\gcd(n, q-1).
\end{eqnarray*}
\end{thm}

\begin{proof}

In Eq. (3.1), by  \cite{MVO} for each divisor $t$ of $m_w$ the number of irreducible binomials of degree $t$ in $\Bbb F_{q^w}[x]$ is
$\frac{\varphi(t)}t\cdot\gcd(n, q^w-1)$.
Moreover,  $x^t-\delta^{ul_w}$ is also irreducible in $\Bbb F_q[x]$ if and only if $\frac{q^w-1}{q-1}| ul_w$ if and only if $\delta^{ul_w}\in \Bbb F_q$. In the following, we investigate it.

 Suppose that $\gcd(n,q-1,\frac{q^w-1}{q-1})=1$. By Lemma 3.1, we have that $n=n_1 n_2$, $rad(n_1)|(q-1)$, $rad(n_2)|\frac{q^w-1}{q-1}$ and $\gcd(n_1,n_2)=1$. So $\gcd(n, q^w-1)=\gcd(n,q-1)\gcd(n, \frac{q^w-1}{q-1})$
and
$$\gcd(\frac{q^w-1}{q-1}, l_w)=\gcd(\frac {q^w-1}{q-1}, \frac {q-1}{\gcd(n, q-1)}\cdot\frac{\frac{q^w-1}{q-1}}{\gcd(n, \frac{q^w-1}{q-1})})=\frac{\frac{q^w-1}{q-1}}{\gcd(n, \frac{q^w-1}{q-1})}.$$
Hence, $x^t-\delta^{ul_w}\in \Bbb F_q[x]$ is irreducible  if and only if  $\frac{q^w-1}{q-1
}| ul_w$  if and only if $\gcd(n, \frac{q^w-1}{q-1})| u$.
We have  $m_w=\frac{n}{\gcd(n,q^w-1)}=\frac{n_{1}}{\gcd(n_{1},q-1)}\cdot \frac{n_{2}}{\gcd(n_{2},\frac{q^w-1}{q-1})}$, and set $$m_{w,1}=\frac{n_{1}}{\gcd(n_1,q-1)}.$$



If $t| m_{w,1}$, we have $x^t-\delta^{ul_w}\in\Bbb F_q[x]$ is also irreducible if and only if $u=\gcd(n_{}, \frac{q^w-1}{q-1})u'$ for $1\le u'\le \gcd(n, q-1)$ and $\gcd(t,u')=1$. Since $\gcd(u',t)=1$ and $rad(t)\mid rad(n_1)\mid \gcd(n,q-1)$, every prime that divides $t$ also divides $\gcd(n,q-1)$. Let $p_{1},p_{2},\cdots,p_{k}$ be the list of primes that divide $t$. It follows that there exist $(1 - \frac{1}{p_{1}} ) \cdot \gcd(q - 1, n)$ numbers less or equal to
$\gcd(q -1, n)$ that do not have any common factor with $p_{1}$.
Inductively we conclude that $(1-\frac{1}{p_{1}})(1-\frac{1}{p_{2}})\cdots (1-\frac{1}{p_{k}})\cdot \gcd(n,q-1)=\frac{\varphi(t)}{t}\gcd(n,q-1)$ numbers without any common factor with $t$.
Hence
there exist $\frac{\varphi(t)}t\cdot
\gcd(n, q-1)$   irreducible binomials  of degree $t$ in $\Bbb F_{q^w}[x]$ that are in $\Bbb F_q[x]$. On the other hand, for each binomial $x^t-\delta^{ul_w}\in \Bbb F_{q^w}[x]$,  $\gcd(n, \frac{q^w-1}{q-1})\nmid u$,  there are $w$  conjugate binomials in $\Bbb F_{q^w}[x]$ such that their product generates an irreducible polynomial in $\Bbb F_q[x]$. Thus the number of irreducible polynomials of degree $wt$ in $\Bbb F_q[x]$ is
$$\frac {\varphi(t)}{wt}\cdot( \gcd(n,q^w-1)- \gcd(n, q-1)).$$

If $t\nmid m_{w,1}$, we have
there is not any binomial polynomial of degree $t$ in $\Bbb F_{q^w}[x]$ that is also in $\Bbb F_q[x]$. Hence the number of irreducible polynomials of degree $wt$ in $\Bbb F_q[x]$ is
$$\frac {\varphi(t)}{wt}\cdot \gcd(n,q^w-1).$$
Hence we have the irreducible factorization of $x^n-1$ over $\Bbb F_q$ in (3.2).

Moreover, observe now that the function $\varphi(t)$ is a multiplicative function, then $\sum_{t\mid m_{w}}\frac {\varphi(t)}{t}$ and $\sum_{t\mid m_{w,1}}\frac {\varphi(t)}{t}$ are also
multiplicative and thus it is enough to calculate this sum for powers of primes. In this case
    we have
\begin{eqnarray*}&&\sum_{d|p^{k}}\frac{\varphi(d)}{d}=1+k(1-\frac 1 p).
\end{eqnarray*}
Finally, the total number of irreducible factors of $x^n-1$ in $\Bbb F_q[x]$ is
\begin{eqnarray*}&&\sum_{t|m_{w,1}}\frac{\varphi(t)}{wt}\cdot(\gcd(n, q^w-1)+(w-1)\gcd(n, q-1))+\sum_{t{\nmid} m_{w,1}}\frac{\varphi(t)}{t}\cdot \frac{\gcd(n, q^w-1)}{w}\\
&=&\sum_{t|m_{w}}\frac{\varphi(t)}{t}\cdot \frac{\gcd(n,q^w-1)}{w}+\sum_{t|m_{w,1}}\frac{\varphi(t)}{t}\cdot \frac{w-1}{w}\cdot\gcd(n, q-1) \\
&=&\prod_{\mbox{\tiny$
\begin{array}{c}
p| m_{w}\\
p~prime\\
\end{array}$
}}(1+v_p(m_{w})\frac {p-1}p)\cdot \frac{\gcd(n, q^w-1)}w\\ &+&\prod_{\mbox{\tiny$
\begin{array}{c}
p| m_{w,1}\\
p~prime\\
\end{array}$
}}(1+v_p(m_{w,1})\frac{p-1}p)\cdot \frac{w-1}{w}\cdot\gcd(n, q-1) .
\end{eqnarray*}

 Suppose that   $\gcd(n, q-1, \frac{q^w-1}{q-1})=w$, where $w$ is an odd prime. Let $n=w^{v_w(n)} n_1 n_2$, where $rad(n_1)|(q-1)$, $rad(n_2)|\frac{q^w-1}{q-1}$, and $\gcd(w, n_1n_2)=1$.

If  $v_w(n)\ge v_w(q^w-1)$, by Lemma 3.1 $\gcd(n, q^w-1)=\gcd( n, q-1)\gcd(n, \frac{q^w-1}{q-1})$ and
\begin{eqnarray*}\gcd(\frac{q^w-1}{q-1}, l_w)&=&\gcd(\frac{q^w-1}{q-1}, \frac{q-1}{\gcd(n, q-1) }\cdot\frac{\frac{q^w-1}{q-1}}{\gcd(n,\frac{q^w-1}{q-1})})=\frac{\frac{q^w-1}{q-1}}{\gcd(n,\frac{q^w-1}{q-1})}.\end{eqnarray*}
We have  $m_w=\frac{n}{\gcd(n,q^w-1)}=w^{v_w(n)-v_w(q^w-1)}\cdot\frac{n_{1}}{\gcd(n_{1},q-1)}\cdot \frac{n_{2}}{\gcd(n_{2},\frac{q^w-1}{q-1})}$, and set $$m_{w,1}=\frac{n_{1}}{\gcd(n_1,q-1)}.$$
Hence, $x^t-\delta^{ul_w}\in \Bbb F_q[x]$ is irreducible  if and only if  $\frac{q^w-1}{q-1
}| ul_w$  if and only if $\gcd(n, \frac{q^w-1}{q-1})| u$.
Similarly, we have the irreducible factorization of $x^n-1$ over $\Bbb F_q$ in (3.2) and
the total number of irreducible factors of $x^n-1$ in $\Bbb F_q[x]$.

If  $v_w(n)< v_w(q^w-1)$,  by Lemma 3.1, $\gcd(n, q^w-1)=\gcd( n/w, q-1)\gcd(n, \frac{q^w-1}{q-1})$
 and
\begin{eqnarray*}\gcd(\frac{q^w-1}{q-1}, l_w)&=&\gcd(\frac{q^w-1}{q-1}, \frac{q-1}{\gcd(n/w, q-1) }\cdot\frac{\frac{q^w-1}{q-1}}{\gcd(n,\frac{q^w-1}{q-1})})=\frac{w\cdot\frac{q^w-1}{q-1}}{\gcd(n,\frac{q^w-1}{q-1})}.\end{eqnarray*}
We have  $m_w=\frac{n}{\gcd(n,q^w-1)}=\frac{n_{1}}{\gcd(n_{1},q-1)}\cdot \frac{n_{2}}{\gcd(n_{2},\frac{q^w-1}{q-1})}$ and set $$m_{w,1}=\frac{n_{1}}{\gcd(n_1,q-1)}.$$

Hence, $x^t-\delta^{ul_w}\in \Bbb F_q[x]$ is irreducible  if and only if  $\frac{q^w-1}{q-1
}| ul_w$  if and only if $\frac{1}{w}\cdot\gcd(n, \frac{q^w-1}{q-1})| u$. Note that $\gcd(n, q^w-1)=\gcd(n, q-1)\cdot\frac 1w\cdot \gcd(n,\frac{q^w-1}{q-1})$.
Similarly, we have the irreducible factorization of $x^n-1$ over $\Bbb F_q$ in (3.2) and
the total number of irreducible factors of $x^n-1$ in $\Bbb F_q[x]$.
\end{proof}

 \begin{exa} Suppose that $w=3$ and $n,q$ are positive integers satisfying the condition in Theorem $3.2$, i.e., $ord_{rad(n)}(q)=3$ and either $q\not\equiv3 \pmod 4$ or $8\nmid n$.
In Table $1$, we list irreducible factor numbers of $x^{n}-1$ over $\Bbb F_{q}$ for $q<10$ by Theorem $3.2$.
 \end{exa}
\[ \begin{tabular} {c} Table $1$. Total irreducible factor number of $x^{n}-1(k_{}\geqslant1)$\\
\begin{tabular}{|c|c|c|c|c|c|c|c|}
  \hline
 $q$ & $n$& $m_{w}$& Irreducible factors number \\
\hline
     $ 2$   &    $7^{k}$&$7^{k-1} $&$2k+1$\\
   \hline
   $3$  &    $2^{k_{1}}13^{k_{}}(0\leqslant k_{1}\leqslant1)$&$13^{k_{}-1} $ &$2^{k_{1}}(4k_{}+1)$\\
   \hline
    $3$  &    $2^{k_{1}}13^{k_{}}(k_{1}\geqslant2)$&$2^{k_{1}-1}13^{k_{}-1}$ &$(k_{1}+1)(4k_{}+1)$\\
   \hline
   $4$  &    $3^{k_{1}}7^{k_{}}(0\leqslant k_{1}\leqslant1)$&$7^{k_{}-1} $&$3^{k_{1}}(2k_{}+1)$\\
   \hline
   $4$  &    $3^{k_{1}}7^{k_{}}(k_{1}\geqslant2)$& $3^{k_{1}-2}7^{k_{}-1}$&$12k_{1}k_{}+2k_{1}-6k_{}+1$\\
   \hline
    $5$  &    $2^{k_{1}}31^{k_{}}(0\leqslant k_{1}\leqslant2)$&$31^{k_{}-1} $ &$2^{k_{1}}(10k_{}+1)$\\
   \hline
      $5$  &    $2^{k_{1}}31^{k_{}}(k_{1}\geqslant2)$&$2^{k_{1}-2}31^{k_{}-1} $&$2k_{1}(10k_{}+1)$\\
   \hline
     $7$  &  $2^{k_{1}}3^{k_{2}}19^{k_{}}(0\leqslant k_{1},k_{2}\leqslant1)$&$19^{k_{}-1} $ &$2^{k_{1}}3^{k_{2}}(6k_{}+1)$\\
    \hline
     $7$  &  $2^{k_{1}}3^{k_{2}}19^{k_{}}(0\leqslant k_{1}\leqslant1,k_{2}\geqslant2)$&$3^{k_{2}-2}19^{k_{}-1} $ &$2^{k_{1}}(36k_{2}k_{}+2k_{2}-18k_{}+1)$\\
    \hline
      $7$  &  $4\times3^{k_{2}}19^{k_{}}(0\leqslant k_{2}\leqslant1)$&$2\times19^{k_{}-1} $ &$3^{k_{2}+1}(6k_{}+1)$\\
    \hline
    $7$  &  $4\times3^{k_{2}}19^{k_{}}(k_{2}\geqslant2)$&$2\times3^{k_{2}-2}19^{k_{}-1} $ &$3(36k_{2}k_{}+2k_{2}-18k_{}+1)$\\
    \hline
    $8$  &  $7^{k_{1}}73^{k_{}}(0\leqslant k_{1}\leqslant1)$&$73^{k_{}-1} $ &$7^{k_{1}}(24k_{}+1)$\\
   \hline
   $8$  &    $7^{k_{1}}73^{k_{}}(k_{1}\geqslant1)$&$7^{k_{1}-1}73^{k_{}-1} $ &$(6k_{1}+1)(24k_{}+1)$\\
   \hline
   $9$ &    $2^{k_{1}}7^{k_{2}}13^{k_{}}(0\leqslant k_{1}\leqslant3,0\leqslant k_{2}\leqslant1)$&$13^{k_{}-1} $ &$\frac{2^{k_{1}}(7^{k_{2}}(12k_{}+1)+2)}{3}$\\
   \hline
     $9$ &    $2^{k_{1}}7^{k_{2}}13^{k_{}}(0\leqslant k_{1}\leqslant3,k_{2}\geqslant2)$&$7^{k_{2}-1}13^{k_{}-1} $ &$2^{k_{1}}(6k_{2}k_{}+2k_{2}+4k_{}+1)$\\
   \hline
        $9$ &    $2^{k_{1}}7^{k_{2}}13^{k_{}}(k_{1}\geqslant4,0\leqslant k_{2}\leqslant1)$&$2^{k_{1}-4}13^{k_{}-1} $ &$\frac{4(k_{1}-2)(7^{k_{2}}(12k_{}+1)+2)}{3}\cdot $\\
   \hline
     $9$ &    $2^{k_{1}}7^{k_{2}}13^{k_{}}(k_{1}\geqslant4,k_{2}\geqslant2)$&$2^{k_{1}-4}7^{k_{2}-1}13^{k_{}-1} $ &$4 (k_{1}-2)(24k_{2}k_{}+2k_{2}+4k_{}+1)$\\
   \hline
\end{tabular}
\end{tabular}
\]

Second, we consider the case: $q\equiv 3\pmod 4$ and $8|n$.  By $w$ an odd prime,  $q^{w} \equiv3 \pmod 4$  if and only if $q \equiv3 \pmod 4$.

Let $m_{2w}=\frac{n}{\gcd(n, q^{2w}-1)}$, $l_{2w}=\frac{q^{2w}-1}{\gcd(n, q^{2w}-1)}$, $l_w=\frac{q^w-1}{\gcd(n, q^w-1)}$, $r=\min\{v_2(n/2),v_2(q^w+1)\}=\min\{v_2(n/2),v_2(q+1)\}$,  $\pi$ a  generator of $\Bbb F_{q^{2w}}^*$, and   $\delta=\pi^{q^w+1}$.
Then the factorization of $x^{n} -1$ into irreducible factors in $\Bbb F_{q^w} [x ]$ is

\begin{equation}\prod_{\mbox{\tiny$
\begin{array}{c}
t|m_{2w}\\
t~odd\\
\end{array}$
}} \prod_{\mbox{\tiny$
\begin{array}{c}
1\leqslant v\leqslant \gcd(n,q^w-1)\\
\gcd(v,t)=1\\
\end{array}$
}}(x^{t}-\delta^{vl_{w}})\cdot\prod\limits_{t\mid m_{2w}} \prod\limits_{u\in \mathcal{R}_{t}} (x^{2t}-(\pi^{ul_{2w}}+\pi^{q^wul_{2w}})x^{t}+\delta^{ul_{2w}}), \end{equation} where
$$\mathcal{R}_t=\bigg\{u\in \Bbb N: \begin{array}{l}
1\le u\le \gcd(n,q^{2w}-1),\gcd(u,t)=1,
 \\ 2^{r}\nmid u, ~and~ u=min\{u, q^wu\}_{\gcd(n,q^{2w}-1)}\end{array}\bigg\}.$$
Note that  $\{a\}_{b}$ denotes the remainder of the division of $a$ by $b$.

  Next we will explicitly give the irreducible factorization of $x^{n}-1$ over $\Bbb F_{q}$.

\begin{thm}  Suppose that $ord_{rad(n)}(q)=w$, where $w$ is an odd prime, $q^{} \equiv 3\pmod 4$  and $8| n$. Let $n=w^{v_w(n)}n_1n_2$, $rad(n_1)|(q-1)$, $rad(n_2)|\frac{q^w-1}{q-1}$, $\gcd(w,n_1n_2)=1$.  Set $m_{2w}=\frac{n}{\gcd(n,q^{2w}-1)}$, $m_{w}=\frac{n}{\gcd(n,q^{w}-1)}$, $m_{w,1}=\frac{n_1}{\gcd(n_{1}, q-1)}$, $l_{2w}=\frac{q^{2w}-1}{\gcd(n, q^{2w}-1)}$, $l_{w}=\frac{q^{w}-1}{\gcd(n,q^{w}-1)}$, $l_2=\frac{q^2-1}{\gcd(n,q^2-1)}$,  $l_{1}=\frac{q^{}-1}{\gcd(n,q^{}-1)}$, $r=min\{v_{2}(\frac n 2),v_{2}(q^{w}+1)\}$, $\pi$ a  generator of $\Bbb F_{q^{2w}}^\ast$,  $\delta=\pi^{q^{w}+1}$, $\alpha=\pi^{\frac{q^{2w}-1}{q^2-1}}$, and $\theta=\delta^{\frac{q^w-1}{q-1}}$. Then

$(1)$ The irreducible factorization of $x^{n}-1$ over $\Bbb F_{q}$ is as follows:
\begin{equation*}
\begin{split}
&\prod_{\mbox{\tiny$
\begin{array}{c}
t| m_{w,1}\\ t~ odd
\end{array}$
}}
 \prod_{\mbox{\tiny$
\begin{array}{c}
1\leqslant v' \leqslant \gcd(n,q-1) \\
\gcd(t,v')=1
\end{array}$
}}(x^{t}-\theta^{v'l_{1}})\cdot\prod_{\mbox{\tiny$
\begin{array}{c}
t|m_{2w} \\
v\in\mathcal{S}_{t}
\end{array}$
}}\prod\limits_{k=0}^{w-1} (x^{t}-\delta^{vq^kl_{w}})\\\\
& \cdot\prod_{\mbox{\tiny$
\begin{array}{c}
t| m_{w,1}\\u'\in \mathcal{R}^{(1)}_t
\end{array}$
}}
(x^{2t}-(\alpha^{u'l_{2}}+\alpha^{u'l_{2}})x^{t}+\theta^{u'l_{2}}) \cdot\prod_{\mbox{\tiny$
\begin{array}{c}
t| m_{2w}\\
u\in\mathcal{ R}^{(2)}_{t}
\end{array}$
}}\prod\limits_{k=0}^{2w-1}(x^t-\pi^{uq^kl_{2w}}),
\end{split}
\end{equation*}
where $$\mathcal{S}_t=\bigg\{v\in \Bbb N: \begin{array}{l}
1\le v\le \gcd(n,q^{w}-1),\gcd(v,t)=1,
 \\ \gcd(n, \frac{q^{w}-1}{q-1})\nmid v, v=\min\{v, qv, \cdots, q^{w-1}v\}_{\gcd(n, q^w-1)}\end{array}\bigg\},$$
$$\mathcal{R}^{(1)}_t=\bigg\{u'\in \Bbb N: \begin{array}{l}
1\le u'\le 2^{r}\gcd(n,q^{}-1),\gcd(u',t)=1,
 \\ 2^r\nmid u', u'=\min\{u',u'q^w\}_{\gcd(n, q^2-1)}\end{array}\bigg\},$$ and
$$\mathcal{R}_t^{(2)}=\bigg\{u\in \Bbb N: \begin{array}{l}
1\le u\le \gcd(n,q^{2w}-1),\gcd(u,t)=1,2^{r}\nmid u,
 \\  \gcd(n, \frac{q^w-1}{q-1})\nmid u,  ~and~ u=\min\{u, uq,\cdots,  q^{2w-1}u\}_{\gcd(n,q^{2w}-1)}\end{array}\bigg\}.$$


$(2)$ For each odd $t$ and $t|m_{w,1}$, the numbers of irreducible polynomials of degree $t$, $2t$,  $wt$,  $2wt$ over  $\Bbb F_{q}$ are $\frac{\varphi(t)}{t}\cdot\gcd(n,q-1)$, $\frac{\varphi(t)}{2t}\cdot(2^{r_{}}-1)\cdot\gcd(n,q-1)$, $\frac{\varphi(t)}{wt}\cdot(\gcd(n,q^{w}-1)-\gcd(n,q-1))$ and $\frac{\varphi(t)}{2wt}\cdot(2^{r_{}}-1)\cdot(\gcd(n,q^{w}-1)-\gcd(n,q-1))$, respectively.
For each odd $t$ and $t\nmid m_{w,1}$, the numbers of irreducible polynomials of degree  $wt$,  $2wt$ over  $\Bbb F_{q}$ are $\frac{\varphi(t)}{wt}\cdot\gcd(n,q^{w}-1)$ and $\frac{\varphi(t)}{2wt}\cdot(2^{r_{}}-1)\cdot\gcd(n,q^{w}-1)$, respectively.
For each even $t$  and $t| m_{w,1}$, the numbers of irreducible polynomials of degree  $2t$,  $2wt$ over  $\Bbb F_{q}$ are $\frac{\varphi(t)}{2t}\cdot2^{r_{}}\cdot\gcd(n,q-1)$ and $\frac{\varphi(t)}{2wt}\cdot2^{r_{}}\cdot(\gcd(n,q^{w}-1)-\gcd(n,q-1))$, respectively.
For each even $t$ and $t\nmid m_{w,1}$, the numbers of irreducible polynomials of degree   $2wt$ over  $\Bbb F_{q}$ is $\frac{\varphi(t)}{2wt}\cdot2^{r_{}}\cdot\gcd(n,q^{w}-1)$.
The total number of irreducible factors of $x^n-1$ in $\Bbb F_q[x]$ is

\begin{eqnarray*}
&&\prod_{\mbox{\tiny$
\begin{array}{c}
p|m_{2w}\\
p~prime
\end{array}$
}}
(1+v_p(m_{2w})\frac {p-1}p)\cdot \frac{2^{r_{}-1}}{w}\cdot\gcd(n,q^{w}-1) \\&+&\prod_{\mbox{\tiny$
\begin{array}{c}
p|m_{2w}\\
p~odd ~prime
\end{array}$
}}
(1+v_p(m_{2w})\frac {p-1}p)\cdot \frac{1}{2w}\cdot\gcd(n,q^{w}-1)\\
&+&\prod_{\mbox{\tiny$
\begin{array}{c}
p|m_{w,1}\\
p~prime
\end{array}$
}}(1+v_p(m_{2w})\frac {p-1}p)\cdot\frac{2^{r_{}-1}(w-1)}{w}\cdot \gcd(n,q-1)\\
&+&\prod_{\mbox{\tiny$
\begin{array}{c}
p|m_{w,1}\\
p~odd~prime
\end{array}$
}}(1+v_p(m_{2w})\frac {p-1}p)\cdot\frac{w-1}{2w}\cdot \gcd(n, q-1).\end{eqnarray*}
\end{thm}

\begin{proof}

 Since $rad(n)|(q^{w}-1)$ and $w$ is odd,  we have $\gcd(n/2, q^{w}+1)=\gcd(n/2, q+1)=2^{r}$, where $r=\min\{v_2(n/2), v_2(q^w+1)\}=\min\{v_2(n/2),v_2(q+1)\}$,  $$m_{2w}=\frac{n}{\gcd(n,q^{2w}-1)}=\frac{n}{\gcd(n/2,q^{w}+1)\gcd(n,q^{w}-1)}=\frac{m_{w}}{2^{r}}$$ and $$l_{2w}=\frac{q^{2w}-1}{\gcd(n,q^{2w}-1)}=\frac{q^{w}-1}{\gcd(n,q^{w}-1)}\cdot \frac{q^{w}+1}{\gcd(n/2,q^{w}+1)}=\frac{q^{w}+1}{2^{r}}l_{w}.$$

 Now  we consider two cases.

 {\bf Case 1.} Suppose that  $\gcd(n,q-1, \frac{q^w-1}{q-1})=1$.
 Let $n=n_{1} n_{2}, rad(n_{1})|(q-1), rad(n_{2})|\frac{q^{w}-1}{q-1}$ and $\gcd(n_{1},n_{2})=1$. By $\gcd(n,q^{w}-1)=\gcd(n,q-1)\gcd(n,\frac {q^w-1}{q-1})$,
 we have $$m_{w}=\frac{n}{\gcd(n,q^{w}-1)}=\frac{n_{1}}{\gcd(n_{1},q^{}-1)}\cdot\frac{n_{2}}{\gcd(n_{2},\frac{q^{w}-1}{q-1})}$$
and set $$m_{w,1}=\frac{n_{1}}{\gcd(n_{1},q-1)}.$$

Firstly, we investigate the product in (3.3):
$$\prod_{\mbox{\tiny$
\begin{array}{c}
t| m_{2w}\\
t~odd\\
\end{array}$
}} \prod_{\mbox{\tiny$
\begin{array}{c}
1\leqslant v\leqslant \gcd(n,q^w-1)\\
\gcd(v,t)=1\\
\end{array}$
}}(x^{t}-\delta^{vl_{w}}).$$
We determine
 when  these binomial polynomials $x^t-\delta^{vl_w}\in \Bbb F_q[x]$ are irreducible.

Note that
$$\gcd(\frac{q^w-1}{q-1}, l_w)=\gcd(\frac {q^w-1}{q-1}, \frac {q-1}{\gcd(n_, q-1)}\cdot\frac{\frac{q^w-1}{q-1}}{\gcd(n_, \frac{q^w-1}{q-1})})=\frac{\frac{q^w-1}{q-1}}{\gcd(n_, \frac{q^w-1}{q-1})}.$$
For $t|m_{2w}$, $t$ odd, $1\le v\le \gcd(n, q^w-1)$ and $\gcd(t, v)=1$,
 binomial polynomial $x^t-\delta^{vl_w}\in \Bbb F_q[x]$ is irreducible if and only if $\gcd(n_, \frac{q^w-1}{q-1})| v$.

If $t| m_{w,1}$ and $t$ is odd, we have $x^t-\delta^{vl_w}\in\Bbb F_q[x]$ is also irreducible if and only if $v=\gcd(n_{2}, \frac{q^w-1}{q-1})v'$ for $1\le v'\le \gcd(n, q-1)$ and $\gcd(t,v')=1$. Hence there exist $\frac{\varphi(t)}t\cdot\gcd(n, q-1)$ binomial polynomials of degree $t$ in $\Bbb F_{q^w}[x]$ that are in $\Bbb F_q[x]$.
On the other hand, for each binomial $x^t-\delta^{vl_w}\in \Bbb F_{q^w}[x]$,  $\gcd(n, \frac{q^w-1}{q-1})\nmid v$,  there are $w$  conjugate binomials in $\Bbb F_{q^w}[x]$ such that their product generates an irreducible polynomial in $\Bbb F_q[x]$,  i.e. $\prod_{k=0}^{w-1}(x^t-\delta^{vq^kl_w})\in \Bbb F_q[x]$ is irreducible. Thus the number of irreducible polynomials of degree $wt$ in $\Bbb F_q[x]$ is
$\frac {\varphi(t)}{wt}\cdot( \gcd(n,q^w-1)- \gcd(n, q-1)).$

If $t\nmid m_{w,1}$ and $t$ is odd, we have
there is not any binomial polynomial of degree $t$ in $\Bbb F_{q^w}[x]$ that is also in $\Bbb F_q[x]$. Hence the number of irreducible polynomials of degree $wt$ in $\Bbb F_q[x]$ is
$\frac {\varphi(t)}{wt}\cdot \gcd(n,q^w-1).$

Secondly, we investigate the product in (3.3):
\begin{equation*}\prod\limits_{t\mid m_{2w}} \prod\limits_{u\in \mathcal{R}_{t}} (x^{2t}-(\pi^{ul_{2w}}+\pi^{q^wul_{2w}})x^{t}+\delta^{ul_{2w}}). \end{equation*}
We determine
when  these polynomials $x^{2t}-(\pi^{ul_{2w}}+\pi^{q^{w}ul_{2w}})x^{t}+\delta^{ul_{2w}}\in \Bbb F_{q}[x]$ are irreducible.

Note that
 $$\gcd(\frac{q^{2w}-1}{q^2-1}, l_{2w})=\frac{q^w+1}{q+1}\gcd(\frac{q^w-1}{q-1}, l_{w})=\frac{\frac{q^{2w}-1}{q^2-1}}{\gcd(n, \frac{q^w-1}{q-1})}. $$

Then irreducible polynomials   $x^{2t}-(\pi^{ul_{2w}}+\pi^{q^{w}ul_{2w}})x^{t}+\delta^{ul_{2w}}=(x^t-\pi^{ul_{2w}})(x^t-\pi^{q^wul_{2w}})\in \Bbb F_{q}[x]$ if and only if $x^t-\pi^{ul_{2w}}\in \Bbb F_{q^2}[x]$ if and only if $\frac{q^{2w}-1}{q^2-1}|ul_{2w}$ if and only if $\gcd(n, \frac{q^w-1}{q-1})|u$ if and only if $u=\gcd(n_{}, \frac{q^w-1}{q-1})u'$ for $1\le u'\le 2^{r}\gcd(n, q-1)$, $\gcd(t,u')=1$ and $2^r\nmid u'$, i.e.  $x^{2t}-(\pi^{ul_{2w}}+\pi^{q^{w}ul_{2w}})x^{t}+\delta^{ul_{2w}}=x^{2t}-(\alpha^{u'l_2}+\alpha^{u'ql_2})x^t+\theta^{u'l_2}\in \Bbb F_q[x]$ is irreducible.

 If $t|m_{w, 1}$ and $t$ is odd,  there exist $\frac1 2\cdot\frac{\varphi(t)}{t}\cdot(2^{r_{}}-1)\cdot\gcd(n,q-1)=\frac{\varphi(t)}{2t}\cdot(2^{r_{}}-1)\cdot\gcd(n,q-1)$ irreducible polynomials of degree $2t$ in $\Bbb F_{q^w}[x]$ that are in $\Bbb F_q[x]$.
  On the other hand, each irreducible polynomial  $x^{2t}-(\pi^{ul_{2w}}+\pi^{q^{w}ul_{2w}})x^{t}+\delta^{ul_{2w}}\in \Bbb F_{q^w}[x]$, $\gcd(n, \frac{q^{w}-1}{q-1})\nmid u$, there are $w$  conjugate irreducible polynomials  in $\Bbb F_{q^w}[x]$ such that their product generates an irreducible polynomial in $\Bbb F_q[x]$, i.e. $\prod_{k=0}^{2w-1}(x^t-\pi^{q^kul_{2w}})\in \Bbb F_q[x]$ is irreducible. From Lemma 2.3, the number of irreducible polynomials of degree $2t$ over $\Bbb F_{q^w}$ in (3.3) is $\frac{\varphi(t)}{2t}\cdot(2^{r_{}}-1)\cdot\gcd(n,q^{w}-1)$. Hence
   the number of irreducible polynomials of degree $2wt$ over  $\Bbb F_{q}$ is $\frac{\varphi(t)}{2wt}\cdot(2^{r_{}}-1)\cdot(\gcd(n,q^{w}-1)-\gcd(n,q-1))$.

  If $t|m_{w, 1}$ and $t$ is even,  we have that $u'$ must be an odd integer. Hence, the number of irreducible polynomials of degree $2t$ and $2wt$ over  $\Bbb F_{q}$ are $\frac{\varphi(t)}{2t}\cdot2^{r_{}}\cdot\gcd(n,q-1)$ and $\frac{\varphi(t)}{2wt}\cdot2^{r_{}}\cdot(\gcd(n,q^{w}-1)-\gcd(n,q-1))$, respectively.

If $t\nmid m_{w,1}$ and $t$ is odd,
there is not any binomial polynomial of degree $2t$ in $\Bbb F_{q^w}[x]$ that is also in $\Bbb F_q[x]$. Hence the number of irreducible polynomials of degree $2wt$ in $\Bbb F_q[x]$ is
$\frac{\varphi(t)}{2wt}\cdot(2^{r_{}}-1)\cdot\gcd(n,q^{w}-1)$.

If $t\nmid m_{w,1}$ and $t$ is even,
there is not any binomial polynomial of degree $2t$ in $\Bbb F_{q^w}[x]$ that is also in $\Bbb F_q[x]$. From Lemma 2.3, the number of irreducible polynomials of degree $2t$ over $\Bbb F_{q^w}$ in (3.3) is $\frac{\varphi(t)}{2t}\cdot2^{r}\cdot\gcd(n,q^{w}-1)$. Hence the number of irreducible polynomials of degree $2wt$ in $\Bbb F_q[x]$ is
$\frac{\varphi(t)}{2wt}\cdot2^{r}\cdot\gcd(n,q^{w}-1)$.

 Hence we have the irreducible factorization of $x^n-1$ over $\Bbb F_q$ in  (1). As each $t$ divides $m_{2w}$ here, the total number of irreducible factors of $x^n-1$ in $\Bbb F_q[x]$ is
\begin{eqnarray*}&&\sum_{\mbox{\tiny$
\begin{array}{c}
t|m_{w,1}\\
t~odd
\end{array}$
}}
\frac{\varphi(t)}{t}\cdot \frac{2^{r_{}}+1}{2w}\cdot(\gcd(n, q^w-1)+(w-1)\gcd(n, q-1))\\
&+&\sum_{\mbox{\tiny$
\begin{array}{c}
t\nmid m_{w,1}\\
t~odd\\
\end{array}$
}}
\frac{\varphi(t)}{t}\cdot \frac{2^{r_{}}+1}{2w}\cdot\gcd(n, q^w-1)\\
&+&\sum_{\mbox{\tiny$
\begin{array}{c}
t| m_{w,1}\\
t~even\\
\end{array}$
}}
\frac{\varphi(t)}{t}\cdot \frac{2^{r_{}}}{2w}\cdot(\gcd(n, q^w-1)+(w-1)\gcd(n,q-1))\\
&+&\sum_{\mbox{\tiny$
\begin{array}{c}
t\nmid m_{w,1}\\
t~even\\
\end{array}$
}}
\frac{\varphi(t)}{t}\cdot\frac{2^{r_{}}}{2w}\cdot\gcd(n,q^{w}-1)\\
&=&\sum_{\mbox{\tiny$
\begin{array}{c}
t |m_{2w}\\
\end{array}$
}}
\frac{\varphi(t)}{t}\cdot\frac{2^{r_{}}}{2w}\cdot\gcd(n,q^{w}-1)
+\sum_{\mbox{\tiny$
\begin{array}{c}
t |m_{2w}\\
t~odd
\end{array}$
}}
\frac{\varphi(t)}{t}\cdot\frac{1}{2w}\cdot\gcd(n,q^{w}-1)\\
&+&\sum_{\mbox{\tiny$
\begin{array}{c}
t| m_{w,1}\\
\end{array}$
}}\frac{\varphi(t)}{t}\cdot\frac{2^{r_{}}(w-1)}{2w}\cdot \gcd(n,q-1)+\sum_{\mbox{\tiny$
\begin{array}{c}
t |m_{w,1}\\
t~odd
\end{array}$
}}\frac{\varphi(t)}{t}\cdot\frac{w-1}{2w}\cdot \gcd(n, q-1) \\
&=&\prod_{\mbox{\tiny$
\begin{array}{c}
p|m_{2w}\\
p~prime
\end{array}$
}}
(1+v_p(m_{2w})\frac {p-1}p)\cdot \frac{2^{r_{}}}{2w}\cdot\gcd(n,q^{w}-1) \\&+&\prod_{\mbox{\tiny$
\begin{array}{c}
p|m_{2w}\\
p~odd ~prime
\end{array}$
}}
(1+v_p(m_{2w})\frac {p-1}p)\cdot \frac{1}{2w}\cdot\gcd(n,q^{w}-1)\\
&+&\prod_{\mbox{\tiny$
\begin{array}{c}
p|m_{w,1}\\
p~prime
\end{array}$
}}(1+v_p(m_{2w})\frac {p-1}p)\cdot\frac{2^{r_{}}(w-1)}{2w}\cdot \gcd(n,q-1)\\
&+&\prod_{\mbox{\tiny$
\begin{array}{c}
p|m_{w,1}\\
p~odd~prime
\end{array}$
}}(1+v_p(m_{2w})\frac {p-1}p)\cdot\frac{w-1}{2w}\cdot \gcd(n, q-1).
\end{eqnarray*}

 {\bf Case 2.} Suppose that  $\gcd(n,q-1, \frac{q^w-1}{q-1})=w$.

 Let $n=w^{v_{w}(n)}n_{1}n_{2}, rad(n_{1})|(q-1), rad(n_{2})|\frac{q^{w}-1}{q-1}$ and $\gcd(w,n_{1}n_{2})=1$.

If $v_{w}(n)\geqslant v_{w}(q^{w}-1)$, by Lemma 3.1, $\gcd(n,q^{w}-1)=\gcd(n,q-1)\gcd(n,\frac{q^{w}-1}{q-1})$.
Then we have  $$m_{w}=\frac{n}{\gcd(n,q^{w}-1)}=w^{v_{w}(n)-v_{w}(q^{w}-1)}\cdot\frac{n_{1}}{\gcd(n_{1},q^{}-1)}\cdot\frac{n_{2}}{\gcd(n_{2},\frac{q^{w}-1}{q-1})}$$
and set $$m_{w,1}=\frac{n_{1}}{\gcd(n_{1},q-1)}.$$

Similarly, each binomial polynomial $x^t-\delta^{vl_w}\in \Bbb F_q[x]$ is irreducible if and only if $v=\gcd(n_{}, \frac{q^w-1}{q-1})v'$ for $1\le v'\le \gcd(n, q-1)$ and $\gcd(t,v')=1$, and each polynomial $x^{2t}-(\pi^{ul_{2w}}+\pi^{q^{w}ul_{2w}})x^{t}+\delta^{ul_{2w}}\in \Bbb F_{q}[x]$ is irreducible if and only if $u=\gcd(n_{}, \frac{q^w-1}{q-1})u'$ for $1\le u'\le 2^{r}\gcd(n, q-1)$, $\gcd(t,u')=1$ and $2^r\nmid u'$. Hence we have the following states.

If $t| m_{w,1}$ and $t$ is odd,  there exist $\frac{\varphi(t)}t\cdot\gcd(n, q-1)$ binomial polynomials of degree $t$ in $\Bbb F_{q^w}[x]$ that are in $\Bbb F_q[x]$.
On the other hand, for each binomial $x^t-\delta^{vl_w}\in \Bbb F_{q^w}[x]$,  $\gcd(n, \frac{q^w-1}{q-1})\nmid v$,  there are $w$  conjugate binomials in $\Bbb F_{q^w}[x]$ such that their product generates an irreducible polynomial in $\Bbb F_q[x]$. Thus the number of irreducible polynomials of degree $wt$ in $\Bbb F_q[x]$ is
$\frac {\varphi(t)}{wt}\cdot( \gcd(n,q^w-1)- \gcd(n, q-1)).$  Also, there exist $\frac{\varphi(t)}{2t}\cdot(2^{r_{}}-1)\cdot\gcd(n,q-1)$ irreducible polynomials of degree $2t$ in $\Bbb F_{q^w}[x]$ that are in $\Bbb F_q[x]$.
  On the other hand, each irreducible polynomial  $x^{2t}-(\pi^{ul_{2w}}+\pi^{q^{w}ul_{2w}})x^{t}+\delta^{ul_{2w}}\in \Bbb F_{q^w}[x]$, $\gcd(n, \frac{q^{w}-1}{q-1})\nmid u$, there are $w$  conjugate irreducible polynomials  in $\Bbb F_{q^w}[x]$ such that their product generates an irreducible polynomial in $\Bbb F_q[x]$.  Hence
   the number of irreducible polynomials of degree $2wt$ over  $\Bbb F_{q}$ is $\frac{\varphi(t)}{2wt}\cdot(2^{r_{}}-1)\cdot(\gcd(n,q^{w}-1)-\gcd(n,q-1))$.

If $t\nmid m_{w,1}$ and $t$ is odd,
there is no any binomial polynomial of degree $t$ in $\Bbb F_{q^w}[x]$ that is also in $\Bbb F_q[x]$. Hence the number of irreducible polynomials of degree $wt$ in $\Bbb F_q[x]$ is
$\frac {\varphi(t)}{wt}\cdot \gcd(n,q^w-1).$ Also,
there is no any binomial polynomial of degree $2t$ in $\Bbb F_{q^w}[x]$ that is also in $\Bbb F_q[x]$. Hence the number of irreducible polynomials of degree $2wt$ in $\Bbb F_q[x]$ is
$\frac{\varphi(t)}{2wt}\cdot(2^{r_{}}-1)\cdot\gcd(n,q^{w}-1)$.

  If $t|m_{w, 1}$ and $t$ is even,  the number of irreducible polynomials of degree $2t$ and $2wt$ over  $\Bbb F_{q}$ are $\frac{\varphi(t)}{2t}\cdot2^{r_{}}\cdot\gcd(n,q-1)$ and $\frac{\varphi(t)}{2wt}\cdot2^{r_{}}\cdot(\gcd(n,q^{w}-1)-\gcd(n,q-1))$, respectively.

If $t\nmid m_{w,1}$ and $t$ is even,
there is no any binomial polynomial of degree $2t$ in $\Bbb F_{q^w}[x]$ that is also in $\Bbb F_q[x]$. Hence the number of irreducible polynomials of degree $2wt$ in $\Bbb F_q[x]$ is
$\frac{\varphi(t)}{2wt}\cdot2^{r_{}}\cdot\gcd(n,q^{w}-1)$.

Similarly, we have the irreducible factorization of $x^n-1$ over $\Bbb F_q$ in  (1) and  the total number of irreducible factors of $x^n-1$ in $\Bbb F_q[x]$.

If $v_{w}(n)< v_{w}(q^{w}-1)$, by Lemma 3.1, $\gcd(n,q^{w}-1)=\gcd(n/w,q-1)\gcd(n,\frac{q^{w}-1}{q-1})$.
Then we have  $$m_{w}=\frac{n}{\gcd(n,q^{w}-1)}=\frac{n_{1}}{\gcd(n_{1},q^{}-1)}\cdot\frac{n_{2}}{\gcd(n_{2},\frac{q^{w}-1}{q-1})}$$
and set $$m_{w,1}=\frac{n_{1}}{\gcd(n_{1},q-1)}.$$
For $t|m_{2w}$, $t$ odd, $1\le v\le \gcd(n, q^w-1)$ and $\gcd(t, v)=1$, we have that $x^t-\delta^{vl_w}\in \Bbb F_q[x]$ is irreducible  if and only if $\frac{\gcd(n, \frac{q^w-1}{q-1})}w\mid v$ as $$\gcd(\frac{q^w-1}{q-1}, l_w)=\gcd(\frac {q^w-1}{q-1}, \frac {q-1}{\gcd(n/w, q-1)}\cdot\frac{\frac{q^w-1}{q-1}}{\gcd(n, \frac{q^w-1}{q-1})})=\frac{w\cdot\frac{q^w-1}{q-1}}{\gcd(n, \frac{q^w-1}{q-1})}.$$ Then we have the following states.

If $t| m_{w,1}$ and $t$ is odd, we have $x^t-\delta^{vl_w}\in\Bbb F_q[x]$ is also irreducible if and only if $v=\frac{\gcd(n_{}, \frac{q^w-1}{q-1})}wv'$ for $1\le v'\le \gcd(n, q-1)$ and $\gcd(t,v')=1$. Hence there exist $\frac{\varphi(t)}t\cdot\gcd(n, q-1)$ binomial polynomials of degree $t$ in $\Bbb F_{q^w}[x]$ that are in $\Bbb F_q[x]$.
On the other hand, for each binomial $x^t-\delta^{vl_w}\in \Bbb F_{q^w}[x]$,  $\frac{\gcd(n_{}, \frac{q^w-1}{q-1})}w\nmid v$,  there are $w$  conjugate binomials in $\Bbb F_{q^w}[x]$ such that their product generates an irreducible polynomial in $\Bbb F_q[x]$,  i.e. $\prod_{k=0}^{w-1}(x^t-\delta^{vq^kl_w})\in \Bbb F_q[x]$ is irreducible. Thus the number of irreducible polynomials of degree $wt$ in $\Bbb F_q[x]$ is
$\frac {\varphi(t)}{wt}\cdot( \gcd(n,q^w-1)- \gcd(n, q-1)).$

If $t\nmid m_{w,1}$ and $t$ is odd,
there is no any binomial polynomial of degree $t$ in $\Bbb F_{q^w}[x]$ that is also in $\Bbb F_q[x]$. Hence the number of irreducible polynomials of degree $wt$ in $\Bbb F_q[x]$ is
$\frac {\varphi(t)}{wt}\cdot \gcd(n,q^w-1).$

Secondly, we investigate whether  each polynomial $x^{2t}-(\pi^{ul_{2w}}+\pi^{q^{w}ul_{2w}})x^{t}+\delta^{ul_{2w}}\in \Bbb F_{q}[x]$ is irreducible in (3.3).

Note that
 $$\gcd(\frac{q^{2w}-1}{q^2-1}, l_{2w})=\frac{q^w+1}{q+1}\gcd(\frac{q^w-1}{q-1}, l_{w})=\frac{w\cdot\frac{q^{2w}-1}{q^2-1}}{\gcd(n, \frac{q^w-1}{q-1})}. $$

Then irreducible polynomials   $x^{2t}-(\pi^{ul_{2w}}+\pi^{q^{w}ul_{2w}})x^{t}+\delta^{ul_{2w}}=(x^t-\pi^{ul_{2w}})(x^t-\pi^{q^wul_{2w}})\in \Bbb F_{q}[x]$ if and only if $x^t-\pi^{ul_{2w}}\in \Bbb F_{q^2}[x]$ if and only if $\frac{q^{2w}-1}{q^2-1}|ul_{2w}$ if and only if $\frac{\gcd(n_{}, \frac{q^w-1}{q-1})}w|u$ if and only if $u=\frac{\gcd(n_{}, \frac{q^w-1}{q-1})}wu'$ for $1\le u'\le 2^{r}\gcd(n, q-1)$, $\gcd(t,u')=1$ and $2^r\nmid u'$, i.e.  $x^{2t}-(\pi^{ul_{2w}}+\pi^{q^{w}ul_{2w}})x^{t}+\delta^{ul_{2w}}=x^{2t}-(\alpha^{u'l_2}+\alpha^{u'ql_2})x^t+\theta^{u'l_2}\in \Bbb F_q[x]$ is irreducible.

 If $t|m_{w, 1}$ and $t$ is odd,  there exist $\frac{\varphi(t)}{2t}\cdot(2^{r_{}}-1)\cdot\gcd(n,q-1)$ irreducible polynomials of degree $2t$ in $\Bbb F_{q^w}[x]$ that are in $\Bbb F_q[x]$.
  On the other hand, each irreducible polynomial  $x^{2t}-(\pi^{ul_{2w}}+\pi^{q^{w}ul_{2w}})x^{t}+\delta^{ul_{2w}}\in \Bbb F_{q^w}[x]$, $\gcd(n, \frac{q^{w}-1}{q-1})\nmid u$, there are $w$  conjugate irreducible polynomials  in $\Bbb F_{q^w}[x]$ such that their product generates an irreducible polynomial in $\Bbb F_q[x]$. Hence
   the number of irreducible polynomials of degree $2wt$ over  $\Bbb F_{q}$ is $\frac{\varphi(t)}{2wt}\cdot(2^{r_{}}-1)\cdot(\gcd(n,q^{w}-1)-\gcd(n,q-1))$.

  If $t|m_{w, 1}$ and $t$ is even,   the number of irreducible polynomials of degree $2t$ and $2wt$ over  $\Bbb F_{q}$ are $\frac{\varphi(t)}{2t}\cdot2^{r_{}}\cdot\gcd(n,q-1)$ and $\frac{\varphi(t)}{2wt}\cdot2^{r_{}}\cdot(\gcd(n,q^{w}-1)-\gcd(n,q-1))$, respectively.

If $t\nmid m_{w,1}$ and $t$ is odd,
there is no any binomial polynomial of degree $2t$ in $\Bbb F_{q^w}[x]$ that is also in $\Bbb F_q[x]$. Hence the number of irreducible polynomials of degree $2wt$ in $\Bbb F_q[x]$ is
$\frac{\varphi(t)}{2wt}\cdot(2^{r_{}}-1)\cdot\gcd(n,q^{w}-1)$.

If $t\nmid m_{w,1}$ and $t$ is even,
there is no any binomial polynomial of degree $2t$ in $\Bbb F_{q^w}[x]$ that is also in $\Bbb F_q[x]$. In this case, by Lemma 2.3, the number of irreducible polynomials of degree $2t$ over $\Bbb F_{q^w}$ in (3.3) is $\frac{\varphi(t)}{2t}\cdot2^{r_{}}\cdot\gcd(n,q^{w}-1)$. Hence the number of irreducible polynomials of degree $2wt$ in $\Bbb F_q[x]$ is
$\frac{\varphi(t)}{2wt}\cdot2^{r_{}}\cdot\gcd(n,q^{w}-1)$.

 Hence we have the irreducible factorization of $x^n-1$ over $\Bbb F_q$ in (1) and the total number of irreducible factors of $x^n-1$ in $\Bbb F_q[x]$.
 \end{proof}

\begin{exa} Suppose that $q=w=3$ and $n=2^{k_{1}}13^{k_{2}}$ where $k_{1}, k_{2}$ are positive integers with $k_{1}\geqslant3, k_{2}\geqslant1$.  It is easy to check that these positive integers satisfying the condition in Theorem $3.4$, i.e., $ord_{rad(n)}(q)=3$ and  $q\equiv3 \pmod 4$ and $8| n$.
By directly calculation, we have $m_{2w}=2^{k_{1}-3}13^{k_{2}-1}$ and $m_{w,1}=2^{k_{1}-3}$.  Hence the total irreducible factors number of $x^{n}-1$ over $\Bbb F_{3}$ is
$8k_{1}k_{2}+2k_{1}-4k_{2}-1$ by Theorem $3.4$. For example, from the formula above, the total irreducible factors number of $x^{104}-1$ over $\Bbb F_{3}$ is $8\times3+2\times3-4-1=25$.


\end{exa}

\subsection{$w=2$.}

\begin{thm}  Assume that  $rad(n)|(q^{2}-1)$ and $rad(n)\nmid (q-1)$. Let $n=2^{v_{2}(n)} n_{1} n_{2}$, $rad(n_{1})|(q-1),  rad(n_{2})|(q+1), \gcd(n_{1},n_{2})=1$.
Set $m_{2}=\frac{n}{\gcd(n,q^{2}-1)}$,  $m_{2,1}=\frac{n_{1}}{\gcd(n_{1},q^{}-1)}$,  $l_{2}=\frac{q^{2}-1}{\gcd(n,q^{2}-1)}$, $l_{1}=\frac{q^{}-1}{\gcd(n,q^{}-1)}$, $\alpha$ a generator of $\Bbb F^{\ast}_{q^{2}}$ satisfying $\alpha^{q+1}=\theta$.
Then

$(1)$ The irreducible factorization of $x^{n}-1$   over $\Bbb F_{q}$ is
\begin{equation}
 \prod_{\mbox{\tiny$
\begin{array}{c}
t|m_{2,1}\\
\end{array}$
}}
 \prod_{\mbox{\tiny$
\begin{array}{c}
1\leqslant u'\leqslant \gcd(n,q-1)\\
\gcd(u',t)=1
\end{array}$
}}(x^{t}-\theta^{u'l_{1}})\cdot\prod_{\mbox{\tiny$
\begin{array}{c}
t| m_{2}\\u\in \mathcal{R}_t\\
\end{array}$
}} (x^{2t}-(\alpha^{ul_{2}}+\alpha^{qul_{2}})x^{t}+\theta^{ul_{2}})\end{equation}
where
$$\mathcal{R}_t=\bigg\{u\in \Bbb N:
\begin{array}{l}
1\le u\le \gcd(n,q^{2}-1),\gcd(u,t)=1,\\
q+1\nmid ul_{2}, u=min\{u,qu\}_{\gcd(n, q^2-1)}\end{array}\bigg\}.$$
$(2)$ For each $t| m_{2,1}$, the number of irreducible binomials  of degree $t$ and degree $2t$ in $\Bbb F_q[x]$ are $\frac{\varphi(t)}t\cdot \gcd(n, q-1)$ and  $\frac{\varphi(t)}{2t}\cdot (\gcd(n, q^{2}-1)-\gcd(n,q-1))$, respectively. For each $t| m_{2}$ with $t\nmid m_{2,1}$,  the number of irreducible binomials  of degree $2t$ in $\Bbb F_q[x]$ is  $\frac{\varphi(t)}{2t}\cdot \gcd(n, q^{2}-1)$.
The total number of irreducible factors of $x^n-1$ in $\Bbb F_q[x]$ is

\begin{eqnarray*}&&\prod_{\mbox{\tiny$
\begin{array}{c}
p| m_{2}\\
p~prime\\
\end{array}$
}}(1+v_p(m_{2})\frac {p-1}p)\cdot \frac{\gcd(n, q^2-1)}2\\&+&\prod_{\mbox{\tiny$
\begin{array}{c}
p| m_{2,1}\\
p~prime\\
\end{array}$
}}(1+v_p(m_{2,1})\frac{p-1}p)\cdot \frac{\gcd(n, q-1)}2.
\end{eqnarray*}
\begin{proof}

By Lemma 2.2, there is a irreducible factorization of $x^n-1$ in $\Bbb F_{q^{2}}[x]$ :

 \begin{equation}
 x^{n} -1=\prod\limits_{t|m_2} \prod_{\mbox{\tiny$
\begin{array}{c}
1\leqslant u\leqslant \gcd(n,q^{2}-1)\\
\gcd(u,t)=1\\
\end{array}$
}}(x^{t}-\alpha^{ul_2}).
\end{equation}

In Eq.(3.5), by  \cite{MVO} for each divisor $t$ of $m_2$ the number of irreducible binomials of degree $t$ in $\Bbb F_{q^2}[x]$ is
$\frac{\varphi(t)}t\cdot\gcd(n, q^2-1)$. Moreover, we know that $x^t-\alpha^{ul_2}$ is also irreducible in $\Bbb F_q[x]$ if and only if $(q+1)|ul_2$ if and only if $\alpha^{ul_2}\in \Bbb F_q$. In the following, we investigate it in the following cases.

{\bf Case 1}. Suppose that $\gcd(n,q-1,q+1)=1$.  Let $n=n_1 n_2$, $rad(n_1)|(q-1)$, $rad(n_2)|(q+1)$ and $\gcd(n_1,n_2)=1$. So $\gcd(n,q^{2}-1)=\gcd(n,q-1)\gcd(n,q+1)$
and
$$\gcd(q+1, l_2)=\gcd(q+1, \frac {q-1}{\gcd(n, q-1)}\cdot\frac{q+1}{\gcd(n, q+1)})=\frac{q+1}{\gcd(n, q+1)}.$$
Hence
we have that $x^t-\alpha^{ul_2}\in \Bbb F_q[x]$ is irreducible  if and only if  $(q+1)| ul_2$  if and only if $\gcd(n, q+1)| u$.
We have  $$m_2=\frac{n}{\gcd(n,q^2-1)}=2^{v_{2}(m_{2})}\cdot\frac{n_{1}}{\gcd(n_{1},q-1)}\cdot \frac{n_{2}}{\gcd(n_{2},q+1)},$$ and set $$m_{2,1}=\frac{n_{1}}{\gcd(n_{1},q-1)}.$$



If $t|m_{w,1}$, then $x^t-\alpha^{ul_2}\in\Bbb F_q[x]$ is also irreducible if and only if $u=\gcd(n_{2}, q+1)u'$ for $1\le u'\le \gcd(n, q-1)$ and $\gcd(t,u')=1$. Hence
there exist $\frac{\varphi(t)}t\cdot
\gcd(n, q-1)$   irreducible binomials  of degree $t$ in $\Bbb F_{q^2}[x]$ that are in $\Bbb F_q[x]$. On the other hand, for each binomial $x^t-\alpha^{ul_2}\in \Bbb F_{q^2}[x]$,  $\gcd(n, q+1)\nmid u$,  there are $2$  conjugate binomials in $\Bbb F_{q^2}[x]$ such that their product generates an irreducible polynomial in $\Bbb F_q[x]$. Thus the number of irreducible polynomials of degree $2t$ in $\Bbb F_q[x]$ is
$\frac {\varphi(t)}{2t}\cdot( \gcd(n,q^2-1)- \gcd(n, q-1)).$

If $t\nmid m_{w,1}$, then
there is no any binomial polynomial of degree $t$ in $\Bbb F_{q^2}[x]$ that is also in $\Bbb F_q[x]$. Hence the number of irreducible polynomials of degree $wt$ in $\Bbb F_q[x]$ is
$\frac {\varphi(t)}{2t}\cdot \gcd(n,q^2-1).$

Hence we have the irreducible factorization of $x^n-1$ over $\Bbb F_q$ in (3.4).
Finally, the total number of irreducible factors of $x^n-1$ in $\Bbb F_q[x]$ is

\begin{eqnarray*}&&\sum_{t|m_{2,1}}\frac{\varphi(t)}{2t}\cdot(\gcd(n, q^2-1)+\gcd(n, q-1))+\sum_{t{\nmid} m_{2,1}}\frac{\varphi(t)}{t}\cdot \frac{\gcd(n, q^2-1)}{2}\\
&=&\prod_{\mbox{\tiny$
\begin{array}{c}
p| m_{2}\\
p~prime\\
\end{array}$
}}(1+v_p(m_{2})\frac {p-1}p)\cdot \frac{\gcd(n, q^2-1)}2\\ &+&\prod_{\mbox{\tiny$
\begin{array}{c}
p| m_{2,1}\\
p~prime\\
\end{array}$
}}(1+v_p(m_{2,1})\frac{p-1}p)\cdot \frac{\gcd(n, q-1)}2 .
\end{eqnarray*}

{\bf Case 2}. Suppose that $\gcd(n,q-1,q+1)=2$. In this case, we have to divided the discussion into two subcases:

Suppose that $q\equiv 1\pmod4$.
Let $n=2^{v_2(n)} n_1 n_2$, where $rad(n_1)|(q-1)$, $rad(n_2)|(q+1)$, and $\gcd(2, n_1n_2)=1$.      Set $m_{2,1}=\frac{n_1}{\gcd(n_1, q-1)}$.

If  $v_2(n)\ge v_2(q-1)+1$,
we have  $\gcd(n, q^2-1)=\gcd( n, q-1)\gcd(n, q+1)$ as $v_{2}(q+1)=1$ and so
\begin{eqnarray*}\gcd(q+1, l_2)&=&\gcd(q+1, \frac{q-1}{\gcd(n, q-1) }\frac{q+1}{\gcd(n,q+1)})=\frac{q+1}{\gcd(n,q+1)}.\end{eqnarray*}

Hence
we have that $x^t-\alpha^{ul_2}\in \Bbb F_q[x]$ is irreducible  if and only if  $(q+1) |ul_2$  if and only if $\gcd(n, q+1)\mid u$.
Similarly, we have the irreducible factorization of $x^n-1$ over $\Bbb F_q$ in (3.4) and
the total number of irreducible factors of $x^n-1$ in $\Bbb F_q[x]$.

If  $v_2(n)\leqslant v_2(q-1)$,  we have $\gcd(n,q^{2}-1)=\gcd(n/2,q-1)\gcd(n,q+1)$
 and
\begin{eqnarray*}\gcd(q+1, l_2)&=&\gcd(q+1, \frac{q-1}{\gcd(n/2, q-1) }\frac{q+1}{\gcd(n,q+1)})=\frac{2( q+1)}{\gcd(n,q+1)}.\end{eqnarray*}

Hence we have that $x^t-\alpha^{ul_2}\in \Bbb F_q[x]$ is irreducible  if and only if  $(q+1)| ul_2$  if and only if $\frac{\gcd(n, q+1)}{2}\mid u$. Note that $\gcd(n, q^2-1)=\gcd(n, q-1)\cdot\frac 12\cdot \gcd(n,q+1)$.
Similarly, we have the irreducible factorization of $x^n-1$ over $\Bbb F_q$ in (3.4) and
the total number of irreducible factors of $x^n-1$ in $\Bbb F_q[x]$.

 Suppose that $q\equiv 3\pmod4$. Let $n=2^{v_2(n)} n_1 n_2$, where $rad(n_1)|(q-1)$, $rad(n_2)|q+1$, and $\gcd(2, n_1n_2)=1$.      Set $m_{w,1}=\frac{n_1}{\gcd(n_1, q-1)}$.

If  $v_2(n)\ge v_2(q+1)+1$,
we have  $\gcd(n, q^2-1)=\gcd( n, q-1)\gcd(n, q+1)$ as $v_{2}(q-1)=1$ and so
\begin{eqnarray*}\gcd(q+1, l_2)&=&\gcd(q+1, \frac{q-1}{\gcd(n, q-1) }\frac{q+1}{\gcd(n,q+1)})=\frac{q+1}{\gcd(n,q+1)}.\end{eqnarray*}

Hence
we have that $x^t-\alpha^{ul_2}\in \Bbb F_q[x]$ is irreducible  if and only if  $(q+1)| ul_2$  if and only if $\gcd(n, q+1)\mid u$.
Similarly, we have the irreducible factorization of $x^n-1$ over $\Bbb F_q$ in (3.4) and
the total number of irreducible factors of $x^n-1$ in $\Bbb F_q[x]$.

If  $v_2(n)\leqslant v_2(q+1)$,  we have $\gcd(n,q^{2}-1)=\gcd(n,q-1)\gcd(n/2,q+1)$
 and
\begin{eqnarray*}\gcd(q+1, l_2)&=&\gcd(q+1, \frac{q-1}{\gcd(n, q-1) }\frac{q+1}{\gcd(n/2,q+1)})=\frac{ q+1}{\gcd(n/2,q+1)}.\end{eqnarray*}

Hence we have that $x^t-\alpha^{ul_2}\in \Bbb F_q[x]$ is irreducible  if and only if  $(q+1)\mid ul_2$  if and only if $\gcd(n/2,q+1)\mid u$.
Similarly, we have the irreducible factorization of $x^n-1$ over $\Bbb F_q$ in (3.4) and
the total number of irreducible factors of $x^n-1$ in $\Bbb F_q[x]$.
\end{proof}

\end{thm}

\begin{exa} Suppose that $w=2$ and $ord_{rad(n)}(q)=2$.
In Table $2$, we list irreducible factor numbers of $x^{n}-1$ over $\Bbb F_{q}$ for $q<10$ by Theorem $3.6$.
\end{exa}
\[ \begin{tabular} {c} Table $2$. Total irreducible factor number of $x^{n}-1$\\
\begin{tabular}{|c|c|c|c|c|c|c|c|}
  \hline
 $q$ & $n$& $m_{w}$&$m_{w,1}$& Irreducible factors number \\
\hline
     $ 2$   &    $3^{k}(k\geqslant1)$&$3^{k-1} $&$1$&$k+1$\\
   \hline
   $4$  &     $3^{k_{1}}5^{k_{2}}(0\leqslant k_{1}\leqslant1, k_{2}\geqslant1)$&$5^{k_{2}-1} $&$1$&$3^{k_{1}}( 2k_{2}+1)$\\
   \hline
    $4$  &     $3^{k_{1}}5^{k_{2}}(k_{1}\geqslant2,k_{2}\geqslant1)$&$3^{k_{1}-1}5^{k_{2}-1} $&$3^{k_{1}-1}$&$ (2k_{1}+1)(2k_{2}+1)$\\
   \hline
    $5$  &    $2^{k_{1}}3^{k_{2}}(0\leqslant k_{1}\leqslant2,k_{2}\geqslant1)$&$3^{k_{2}-1} $ &$1$&$2^{k_{1}}(k_{2}+1)$\\
   \hline
    $ 5$  &    $2^{k_{1}}3^{k_{2}}(k_{1}\geqslant3,k_{2}\geqslant1 )$&$2^{k_{1}-3}3^{k_{2}-1} $ &$1$&$4k_{1}k_{2}+2k_{1}-4k_{2}$\\
   \hline

  $ 8$ &    $3^{k_{1}}7^{k_{2}}(1\leqslant k_{1}\leqslant2,0\leqslant k_{2}\leqslant1 )$&$1$&$1 $&$\frac{ (3^{k_{1}}+1)7^{k_{2}}}{2}$\\
   \hline
    $ 8$ &    $3^{k_{1}}7^{k_{2}}(1\leqslant k_{1}\leqslant2,k_{2}\geqslant2 )$&$7^{k_{2}-1} $&$7^{k_{2}-1}  $&$\frac{7+3^{k_{1}}(6k_{2}+1)}{2} $\\
   \hline
     $ 8$ &    $3^{k_{1}}7^{k_{2}}( k_{1}\geqslant3,0\leqslant k_{2}\leqslant1)$&$3^{k_{1}-2} $ &$1 $&$7^{k_{2}}( 3k_{1}-1)$\\
   \hline
    $ 8$ &    $3^{k_{1}}7^{k_{2}}( k_{1}\geqslant3,k_{2}\geqslant2)$&$3^{k_{1}-2}7^{k_{2}-1} $ &$7^{k_{2}-1} $&$(3k_{1}-1)(6k_{2}+1)$\\
   \hline
     $9$ &    $2^{k_{1}}5^{k_{2}}(0\leqslant k_{1}\leqslant3,k_{2}\geqslant1 )$&$5^{k_{2}-1} $ &$1$&$2^{k_{1}}(2k_{2}+1)$\\
   \hline
   $9$ &    $2^{k_{1}}5^{k_{2}}(k_{1}\geqslant4,k_{2}\geqslant1 )$&$2^{k_{1}-4}5^{k_{2}-1} $ &$1$&$4(4k_{1}k_{2}+k_{1}-8k_{2}-1)$\\
   \hline
\end{tabular}
\end{tabular}
\]

\section{Concluding remarks}
 In this paper, suppose that $rad(n)\nmid (q-1)$ and $rad(n)|(q^w-1)$, where $w$ is prime,  we explicitly give the factorization of $x^{n}-1$ over $\Bbb F_{q}$. If $w$ is a composite number, can we give the explicit factorization of $x^n-1$ over $\Bbb F_q$? We leave this for future work.



\begin{thebibliography}{[kkk]}


  \bibitem{BGM}  I.F.  Blake, S. Gao, R.C. Mullin,  Explicit factorization of $x^{2^k }+1$  over  $\Bbb F_{p}$ with prime $p \equiv3 \pmod 4$, Appl. Algebra Engrg. Comm. Comput. 4 (2) (1993) 89-94.
 \bibitem{CL} B. Chen, H. Liu, and G.Zhang, A class of minimal cyclic codes over finite fields, Des. Codes Cryptogr. 74 (2) (2015)   285-300.
  \bibitem{CLT} B. Chen, L. Li, R. Tuerhong,  Explicit factorization of $X^{2^{m}p^{n}}-1$ over a finite field, Finite Fields Appl. 24 (2013)  95-104.




















  \bibitem{F}  R.W. Fitzgerald, J.L. Yucas, Generalized reciprocals, factors of Dickson polynomials and generalized cyclotomic polynomials over finite fields, Finite Fields Appl. 13 (2007) 492-515.

\bibitem{LY1} F. Li , Q. Yue, C. Li,  Irreducible cyclic codes of length $4p^{n}$ and $8p^{n}$,  Finite Fields Appl. 34 ( 2015) 208-234.


\bibitem{LY2} F. Li , Q. Yue,  The primitive idempotents and weight distributions of irreducible constacyclic codes, Des. Codes Cryptogr.  (2017) https://doi.org/10.1007/s10623-017-0356-2.


 \bibitem{LN} R. Lidl and H. Niederreiter, Finite Fields, Cambridge University Press, Cambridge, 2008.



     \bibitem{MVO} F. Mart\'{\i}nez, C. Vergara, L. Oliveira, Explicit factorization of $x^{n} -1 \in \Bbb F_{q} [x]$, Des. Codes Cryptogr.  77 (1) (2015)  277-286.




   \bibitem{TW} A. Tuxanidy, Q. Wang, Composed products and factors of cyclotomic polynomials over finite fields, Des. Codes Cryptogr. 69 (2) (2013) 203-231.










  \bibitem{WW} L. Wang, Q. Wang, On explicit factors of cyclotomic polynomials over finite fields, Des. Codes Cryptogr. 63 (1) (2012) 87-104.

 \bibitem{W} Q. Wang,  Cyclotomy and permutation polynomials of large indices, Finite Fields Appl. 22 (7) (2013) 57-69.

 \bibitem{WY} Q. Wang, J. L. Yucas,  Dickson polynomials over finite fields,  Finite Fields Appl. 18 (4) (2012) 814-831.

  \bibitem{WZF} H. Wu, L. Zhu, R. Feng , et al, Explicit factorizations of cyclotomic polynomials over finite fields, Des. Codes Cryptogr. 83 (1) (2016) 197-217.










 \end{thebibliography}
\end{document}